\documentclass{llncs}
\usepackage{multicol}

\usepackage{amsmath}
\usepackage{amsfonts}
\usepackage{amssymb}
\usepackage{verbatim}

\usepackage{graphicx}
\usepackage{epic}
\usepackage{eepic}
\usepackage{epsfig,float}

\usepackage{multicol}
\renewcommand{\le}{\leqslant}
\renewcommand{\ge}{\geqslant}

\newcommand{\ol}{\overline}
\newcommand{\eps}{\varepsilon}
\newcommand{\emp}{\emptyset}

\newcommand{\Sig}{\Sigma}

\newcommand{\noin}{\noindent}

\newcommand{\bi}{\begin{itemize}}
\newcommand{\ei}{\end{itemize}}
\newcommand{\be}{\begin{enumerate}}
\newcommand{\ee}{\end{enumerate}}
\newcommand{\bd}{\begin{description}}
\newcommand{\ed}{\end{description}}
\newcommand{\bq}{\begin{quote}}
\newcommand{\eq}{\end{quote}}

\newcommand{\der}[1]{\xrightarrow{#1}}

\newcommand{\cD}{{\mathcal D}}

\newcommand{\cN}{{\mathcal N}}
\newcommand{\cP}{{\mathcal P}}

\newcommand{\cS}{{\mathcal S}}
\newcommand{\cT}{{\mathcal T}}
\newcommand{\cU}{{\mathcal U}}

\newcommand{\cW}{{\mathcal W}}
\newcommand{\cX}{{\mathcal X}}
\newcommand{\one}{{\mathbf 1}}

\newcommand{\bs}{\backslash}

\spnewtheorem{conj}{Conjecture}{\bfseries}{\rmfamily}

\title{Universal Witnesses for State Complexity of  Boolean Operations and Concatenation Combined with Star\thanks{This work was supported by the Natural Sciences and Engineering Research Council of Canada under grant No.~OGP0000871.}
}
\author{Janusz~Brzozowski and David Liu}

\authorrunning{Brzozowski, Liu}   

\institute{David R. Cheriton School of Computer Science, University of Waterloo, \\
Waterloo, ON, Canada N2L 3G1\\
{\tt \{brzozo,dyliu\}@uwaterloo.ca}
}
\begin{document}
\maketitle
\begin{abstract}
We study the state complexity of boolean operations and  product (concatenation, catenation) combined with star. 
We derive tight upper bounds for the symmetric differences and differences of two languages,
one or both of which are starred, and for the product  of two starred languages.
We prove that the previously discovered bounds for the union and the intersection of  languages with one or two starred arguments, for the product of two languages one of which is starred, and for the star of the product of two languages can all be met by the recently introduced
universal witnesses and their variants. 
\smallskip

\noin
\end{abstract}

\section{Introduction}
\label{sec:intr}
The \emph{state complexity of a regular language} is the number of states in the minimal deterministic finite automaton (DFA) recognizing the language.
The \emph{state complexity of an operation} on regular languages is the worst-case state complexity of the result of the operation as a function of the state complexities of the arguments.
For more information on this topic see~\cite{Brz10,Brz12,Yu01}.

 Let $K$ and $L$ be two regular languages over alphabet $\Sig$, and let their state complexities be $m$ and $n$, respectively.
 In 2007 A.\ Salomaa, K.\ Salomaa, and  Yu~\cite{SSY07} showed using ternary witnesses that the complexity of 
$(K\cup L)^*$ is $2^{m+n-1} -(2^{m-1}+2^{n-1}-1)$. 
They also established a lower bound for $(K\cap L)^*$ using an alphabet of 8 letters.
These results were improved by Jir\'askov\'a and Okhotin~\cite{JiOk11} who showed that binary witnesses  suffice for $(K\cup L)^*$, and that $3\cdot 2^{mn-2}$ is a tight upper bound for $(K\cap L)^*$; they used an alphabet of 6 letters.
In 2012, Gao and Yu~\cite{GY_FI12}  showed with ternary witnesses that the complexity of $K\cup L^*$ is 
$m(2^{n-1}+2^{n-2}-1)+1$, and that the same upper bound applies to $K\cap L^*$.
Moreover, it was shown in~\cite{GKY_TCS12} by  Gao, Kari and Yu  that
quaternary witnesses meet  the bound
$(2^{m-1}+2^{m-2}-1)(2^{n-1}+2^{n-2}-1)+1$
for $K^*\cup L^*$ and 
$K^*\cap L^*$.
In 2008, Gao, K. Salomaa, and Yu~\cite{GSY_FI08} demonstrated using quaternary witnesses that 
$2^{m+n-1}+2^{m+n-4}-(2^{m-1}+2^{n-1}-m-1)$ is a tight upper bound for $(KL)^*$. 
The complexity of $KL^*$ was studied by Cui, Gao, Kari and Yu~\cite{CGKY_IJFCS12} in 2012. They proved with ternary witnesses that 
the tight bound is $m(2^{n-1}+2^{n-2})-2^{n-2}$.
The same authors also showed in~\cite{CGKY_TCS12b} using quaternary witnesses  that the complexity of $K^*L$ is
$5\cdot 2^{m+n-3}-2^{m-1}-2^n+1$.
In summary, nine operations using union, intersection, and  product (also called concatenation or catenation) combined with star have been studied.

To establish  the state complexity of an operation
one finds an upper bound and  languages to act as \emph{witnesses} to show that the  bound is tight. 
A witness is usually a sequence $(L_n\mid n\ge k)$ of languages, where  $k$ is some small positive integer; we will call such a sequence a  \emph{stream of languages}.
The languages  in a stream normally differ only in the parameter $n$. In the past,
two different streams have been used for most binary operations.

Recently, Brzozowski~\cite{Brz12} proposed the DFA 
$\cU_n(a,b,c)=(Q,\Sig,\delta,0, \{n-1\})$ of Fig.~\ref{fig:witness} and its language $U_n(a,b,c)$ as the ``universal witness'' DFA and language, respectively, for $n\ge 3$.
The restrictions of the DFA and  the language to alphabet $\{a,b\}$ are denoted by
$\cU_n(a,b,\emp)$ and  $U_n(a,b,\emp)$.
It was proved in~\cite{Brz12} that the bound $2^{n-1}+2^{n-2}$ for star is met by $U_n(a,b,\emp)$, and the bound $2^n$ for reversal, by $U_n(a,b,c)$.
The bound
$(m-1)2^n+2^{n-1}$ for product is met by $U_m(a,b,c)$ and $U_n(a,b,c)$.
The bound $mn$ for union, intersection, difference ($K\setminus L$) and symmetric difference
($K\oplus L$) is met by the streams $U_m(a,b,c)$ and $U_n(a,b,c)$ if $m\neq n$, as was conjectured in~\cite{Brz12} and proved in~\cite{BrLiu12a}.
If $m=n$, it is necessary to use two different streams; however, it is possible to use streams that are almost the same, in the following sense.
Two languages $K$ and $L$ over $\Sig$ are \emph{permutationally equivalent} if one can be obtained from the other by permuting the letters of the alphabet, and a similar definition applies to DFA's.
It was proved in~\cite{Brz12} that
two permutationally equivalent streams 
$U_m(a,b,c)$ and $U_n(b,a,c)$ are witnesses to the bound for the boolean operations:
union ($K\cup L$), intersection ($K\cap L$), difference  ($K\setminus L$), and symmetric difference ($K\oplus L$).
Thus $U_n(a,b,c)$ is indeed a universal witness for the basic operations.

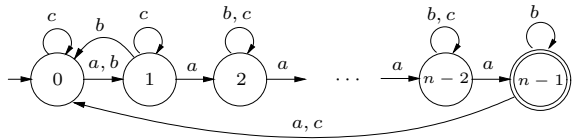
\begin{figure}[h]
\begin{center}
\setlength{\unitlength}{0.00043745in}
\begingroup\makeatletter\ifx\SetFigFont\undefined%
\gdef\SetFigFont#1#2#3#4#5{%
  \reset@font\fontsize{#1}{#2pt}%
  \fontfamily{#3}\fontseries{#4}\fontshape{#5}%
  \selectfont}%
\fi\endgroup%
{\renewcommand{\dashlinestretch}{30}
\begin{picture}(6677,1617)(0,-10)
\put(469,1401){\makebox(0,0)[lb]{\smash{{\SetFigFont{7}{8.4}{\familydefault}{\mddefault}{\updefault}$c$}}}}
\put(2747.000,1151.333){\arc{333.333}{2.2143}{7.2105}}
\blacken\path(2884.638,1114.417)(2847.000,1018.000)(2925.107,1085.913)(2884.638,1114.417)
\put(1667.000,1151.333){\arc{333.333}{2.2143}{7.2105}}
\blacken\path(1804.638,1114.417)(1767.000,1018.000)(1845.107,1085.913)(1804.638,1114.417)
\put(5179.000,1158.333){\arc{333.333}{2.2143}{7.2105}}
\blacken\path(5316.638,1121.417)(5279.000,1025.000)(5357.107,1092.913)(5316.638,1121.417)
\put(6312.000,1211.333){\arc{333.333}{2.2143}{7.2105}}
\blacken\path(6449.638,1174.417)(6412.000,1078.000)(6490.107,1145.913)(6449.638,1174.417)
\put(591,703){\ellipse{630}{630}}
\put(2755,703){\ellipse{630}{630}}
\put(1685,701){\ellipse{630}{630}}
\put(5187,703){\ellipse{630}{630}}
\put(6309,696){\ellipse{720}{720}}
\put(6309,697){\ellipse{630}{630}}
\path(12,703)(282,703)
\blacken\path(162.000,673.000)(282.000,703.000)(162.000,733.000)(162.000,673.000)
\path(1992,703)(2442,703)
\blacken\path(2322.000,673.000)(2442.000,703.000)(2322.000,733.000)(2322.000,673.000)
\path(3072,703)(3522,703)
\blacken\path(3402.000,673.000)(3522.000,703.000)(3402.000,733.000)(3402.000,673.000)
\path(912,703)(1362,703)
\blacken\path(1242.000,673.000)(1362.000,703.000)(1242.000,733.000)(1242.000,673.000)
\path(4422,711)(4872,711)
\blacken\path(4752.000,681.000)(4872.000,711.000)(4752.000,741.000)(4752.000,681.000)
\path(5495,703)(5945,703)
\blacken\path(5825.000,673.000)(5945.000,703.000)(5825.000,733.000)(5825.000,673.000)
\path(1497,965)(1496,967)(1493,970)
	(1488,977)(1480,986)(1469,998)
	(1456,1014)(1441,1031)(1424,1049)
	(1405,1068)(1385,1088)(1364,1106)
	(1341,1125)(1318,1141)(1292,1157)
	(1265,1170)(1236,1182)(1205,1191)
	(1172,1196)(1137,1198)(1103,1195)
	(1070,1188)(1039,1178)(1011,1165)
	(986,1150)(962,1134)(940,1116)
	(919,1097)(900,1077)(881,1057)
	(864,1037)(849,1018)(835,1000)
	(824,985)(814,972)(799,950)
\blacken\path(841.814,1066.047)(799.000,950.000)(891.387,1032.247)(841.814,1066.047)
\path(5997,493)(5996,493)(5995,492)
	(5993,491)(5989,489)(5983,487)
	(5976,483)(5966,479)(5954,474)
	(5940,468)(5923,460)(5903,452)
	(5881,442)(5856,432)(5829,420)
	(5799,408)(5766,395)(5731,381)
	(5695,367)(5656,352)(5615,337)
	(5572,321)(5527,305)(5481,289)
	(5433,273)(5383,257)(5332,241)
	(5280,226)(5225,210)(5169,195)
	(5111,180)(5051,165)(4989,151)
	(4925,137)(4859,123)(4789,110)
	(4717,98)(4642,86)(4564,75)
	(4482,64)(4397,54)(4308,45)
	(4215,37)(4118,30)(4019,24)
	(3915,19)(3810,15)(3702,13)
	(3605,12)(3509,13)(3412,15)
	(3317,17)(3223,21)(3131,25)
	(3041,30)(2952,37)(2866,43)
	(2782,51)(2699,59)(2619,67)
	(2540,76)(2463,86)(2387,96)
	(2313,106)(2240,117)(2169,128)
	(2099,139)(2030,151)(1962,163)
	(1895,175)(1829,188)(1764,200)
	(1701,213)(1638,226)(1577,239)
	(1517,252)(1459,265)(1402,277)
	(1347,290)(1294,302)(1242,314)
	(1194,326)(1147,337)(1103,348)
	(1062,358)(1023,367)(988,376)
	(956,384)(926,391)(900,398)
	(877,404)(858,409)(841,413)
	(827,417)(816,420)(807,422)
	(801,424)(792,426)
\blacken\path(915.650,429.254)(792.000,426.000)(902.635,370.683)(915.650,429.254)
\put(536,636){\makebox(0,0)[lb]{\smash{{\SetFigFont{7}{8.4}{\rmdefault}{\mddefault}{\updefault}$0$}}}}
\put(1624,636){\makebox(0,0)[lb]{\smash{{\SetFigFont{7}{8.4}{\rmdefault}{\mddefault}{\updefault}$1$}}}}
\put(2711,636){\makebox(0,0)[lb]{\smash{{\SetFigFont{7}{8.4}{\rmdefault}{\mddefault}{\updefault}$2$}}}}
\put(2120,801){\makebox(0,0)[lb]{\smash{{\SetFigFont{7}{8.4}{\familydefault}{\mddefault}{\updefault}$a$}}}}
\put(3192,823){\makebox(0,0)[lb]{\smash{{\SetFigFont{7}{8.4}{\familydefault}{\mddefault}{\updefault}$a$}}}}
\put(5622,823){\makebox(0,0)[lb]{\smash{{\SetFigFont{7}{8.4}{\familydefault}{\mddefault}{\updefault}$a$}}}}
\put(957,838){\makebox(0,0)[lb]{\smash{{\SetFigFont{7}{8.4}{\familydefault}{\mddefault}{\updefault}$a,b$}}}}
\put(4917,658){\makebox(0,0)[lb]{\smash{{\SetFigFont{6}{7.2}{\familydefault}{\mddefault}{\updefault}$n-2$}}}}
\put(3875,644){\makebox(0,0)[lb]{\smash{{\SetFigFont{7}{8.4}{\familydefault}{\mddefault}{\updefault}$\cdots$}}}}
\put(4550,816){\makebox(0,0)[lb]{\smash{{\SetFigFont{7}{8.4}{\familydefault}{\mddefault}{\updefault}$a$}}}}
\put(6184,1431){\makebox(0,0)[lb]{\smash{{\SetFigFont{7}{8.4}{\familydefault}{\mddefault}{\updefault}$b$}}}}
\put(3364,133){\makebox(0,0)[lb]{\smash{{\SetFigFont{7}{8.4}{\familydefault}{\mddefault}{\updefault}$a,c$}}}}
\put(6040,640){\makebox(0,0)[lb]{\smash{{\SetFigFont{6}{7.2}{\familydefault}{\mddefault}{\updefault}$n-1$}}}}
\put(4962,1408){\makebox(0,0)[lb]{\smash{{\SetFigFont{7}{8.4}{\familydefault}{\mddefault}{\updefault}$b,c$}}}}
\put(1039,1273){\makebox(0,0)[lb]{\smash{{\SetFigFont{7}{8.4}{\familydefault}{\mddefault}{\updefault}$b$}}}}
\put(2547,1423){\makebox(0,0)[lb]{\smash{{\SetFigFont{7}{8.4}{\familydefault}{\mddefault}{\updefault}$b,c$}}}}
\put(1557,1423){\makebox(0,0)[lb]{\smash{{\SetFigFont{7}{8.4}{\familydefault}{\mddefault}{\updefault}$c$}}}}
\put(592.000,1136.333){\arc{333.333}{2.2143}{7.2105}}
\blacken\path(729.638,1099.417)(692.000,1003.000)(770.107,1070.913)(729.638,1099.417)
\end{picture}
}
\end{center}
\caption{DFA $\cU_n(a,b,c)$ of language $U_n(a,b,c)$.} 
\label{fig:witness}
\end{figure}

It turns out that the witness $U_n(a,b,c)$ cannot meet the bound for some combined operations. However, the notion of universal witness can be broadened to include 
 ``dialects'' of $U_n(a,b,c)$.
Some terminology is required, before we define this concept. 

The inputs of  DFA $\cU_n$ perform the following transformations on the set 
$Q=\{0,\ldots,n-1\}$ of states.
Input $a$ is a \emph{cycle} of all $n$ states, and this is denoted by $a:(0,\ldots,n-1)$.
Input $b$ is a \emph{transposition} of  0 and 1, and does not affect any other states; this is denoted by
$b:(0,1)$, and by $b:(i,j)$, if $i$ and $j$ are transposed. 
Input $c$ is a \emph{singular} transformation sending state $n-1$ to state 0, and not affecting any other states; this is denoted by $c:{n-1\choose 0}$, and by $c:{i\choose j}$, in general.
The \emph{constant} transformation sending all states to state $i$ is denoted by ${Q\choose i}$.
The \emph{identity} transformation on $Q$ is denoted by $\one_Q$.

It is known~\cite{Brz12}  that the inputs of $\cU_n(a,b,c)$ of Fig.~\ref{fig:witness} perform all $n^n$ transformations of states.

A \emph{dialect} of $U_n(a,b,c)$ is the language of any DFA with three inputs $a$, $b$, and $c$, where $a$ is a cycle of length $n$ as above, $b$ is the transposition of \emph{any} two states $(i,j)$, and 
$c$ is a singular transformation $c:{i \choose j}$ sending \emph{any} state $i$ to \emph{any} state $j$.
The initial state is always 0, but the set of final states is arbitrary, as long as the resulting DFA is minimal.

Since there are operations for which ternary witnesses do not meet the worst-case bounds, 
the notions of universal witness and dialect have been extended to quaternary alphabets~\cite{Brz12}, by adding a fourth input $d$ which performs the identity permutation, denoted by $d:\one_Q$.
The concepts of permutational equivalence and dialects were extended in the obvious way to quaternary languages and DFA's.
The following dialects are used in this paper: 
\be
\item
$\cU_{\{0\},n}(a,b,c)$, which is $\cU_n(a,b,c)$ with $\{0\}$ as the set of final states.
\item
$\cT_n(a,b,c)=(Q,\Sig, \delta_T, 0, \{n-1 \})$, where $a:(0,\ldots,n-1)$, $b:(0,1)$, and $c:{1 \choose 0}$.
\item
$\cW_n(a,b,c,d)=(Q,\Sig,\delta_\cW,0,\{n-1\})$,
where $a:(0,\ldots,n-1)$, $b:(n-2,n-1)$, $c:{1 \choose 0}$, and $d: \one_Q$.
\item
$\cW_{\{0\},n}(a,b,c,d)$, which is $\cW_n(a,b,c,d)$ with $\{0\}$ as the set of final states.
\ee
We use the convention that $\cX$ is a DFA if and only if  $X$ is its language.
The operation $K\circ L$ represents any one of the four boolean operations
union, intersection, difference and symmetric difference.

In this paper, we consider the following 13 operations that use boolean operations and product combined with star :

\noin \mbox{}
$\hspace{2cm} K \cup L^*, K \cap L^*, K \oplus L^*, K \bs L^*, L^* \bs K$, \\ \mbox{}
$\hspace{2cm}  K^* \cup L^*, K^* \cap L^*, K^* \oplus L^*, K^* \bs L^*$,\\ \mbox{}
$\hspace{2cm} KL^*, K^*L, K^*L^*, (KL)^*$.
\smallskip

Our contributions are as follows:
\be
\item
We derive the bound $m(2^{n-1}+2^{n-2}-1)+1$ for $K_m\setminus L_n^*$, $L_n^*\setminus K_m$ and $K_m\oplus L_n^*$.  We  show that the known bounds for $K_m\cup L_n^*$,
$K_m\oplus L_n^*$ and $L_n^*\setminus K_m$ are met by the streams $U_m(a,b,c)$ and $U_n(b,a,c)$, and that, for $K_m\cup L_n^*$ and $K_m\setminus L_n^*$, the dialect $U_{\{0\},m}(a,b,c)$  and the language $U_n(b,a,c)$ act as witnesses. 
This corrects an error in~\cite{GY_FI12}, where it is claimed that the witnesses that serve for union also work for intersection.
\item
We derive the bound $(2^{m-1}+2^{m-2}-1)(2^{n-1}+2^{n-2}-1)+1$ for $K_m^*\setminus L_n^*$,  and $K_m^*\oplus L_n^*$.  We  show that the known bounds for $K_m^*\cup L_n^*$ and
$K_m^*\cap L_n^*$  are met by the dialects $W_m(a,b,c,d)$ and $W_n(d,c,b,a)$, and that, for $K_m^*\setminus L_n^*$ and $K_m^*\oplus L_n^*$, the dialects $W_{\{0\},m}(a,b,c,d)$ and  $W_n(d,c,b,a)$ act as witnesses. 

\item
We prove that the known bound $m(2^{n-1}+2^{n-2})-2^{n-2}$ for 
$K_m L_n^*$ is met by the  dialects $T_m(a,b,c)$ and $T_n(b,a,c)$.
\item
We show that the known bound $5\cdot 2^{m+n-3}-2^{m-1}-2^n+1$ for 
$K_m^* L_n$ is met by $U_m(a,b,c,d)$
and $U_n(d,c,b,a)$.
\item
We derive the bound $2^{m+n-1} - 2^{m-1} - 3 \cdot 2^{n-2} + 2$ for 
$K_m^*L_n^*$ and show that it is met by  $U_m(a,b,c,d)$
and $U_n(d,c,b,a)$.
\item
We prove that the known bound $2^{m+n-1}+2^{m+n-4}-(2^{m-1}+2^{n-1}-m-1)$ for $(K_mL_n)^*$ is met by  $W_m(a,b,c,d)$  and $W_n(d,c,b,a)$. 
\item 
In obtaining these results, we prove Conjectures 7, 9, 10, 12, 15 and 17 of~\cite{Brz12}.
\ee

Sections~\ref{sec:onerev} and~\ref{sec:tworev} study boolean operations with one and two starred arguments, respectively. 
Products with one or two starred arguments are examined in
Section~\ref{sec:prodstar}.
In Section~\ref{sec:starops} we consider stars of product, intersection, and difference, and Section~\ref{sec:conc} concludes the paper.

\section{Boolean Operations with One Starred Argument}
\label{sec:onerev}
Recall that the complexity of $L_n^*$ is $2^{n-1}+2^{n-2}$.
Gao and Yu~\cite{GY_FI12}  showed that the complexity of $K_m\cup L_n^*$
is $m(2^{n-1}+2^{n-2}-1)+1$.
They used the following  DFA's over alphabet $\Sig=\{a,b,c\}$:
For $K$, let $\cD_K=(Q_K,\Sig,\delta_K,0,\{m-1\})$, with $Q_K=\{0,\ldots,m-1\}$, 
$a,b:\one_{Q_K}$, and
$c:(0,\ldots n-1)$.
For $L$, let
$\cD_L=(Q_L,\Sig,\delta_L,0,\{n-1\})$, with $Q_L=\{0,\ldots,n-1\}$,
$a:(0,\ldots,{n-1})$,
$b$ defined by $\delta_L(0,b)=0$, $\delta_L(i,b)=i+1 \pmod n $, for $i=1,\ldots,n-1$, and 
$c:\one_{Q_L}$.
They showed that the same bound also holds for $K_m\cap L_n^*$, and  claimed that the same witnesses
work. That claim is incorrect, however, as is shown below.

The results of \cite{GY_FI12} for union are extended here to $K_m \cup L_n^*$, $K_m\oplus L_n^*$ and $L_n^*\setminus K_m$ with witnesses $U_m(a,b,c)$ and $U_n(b,a,c)$, and to
$K_m\cap L_n^*$ and $K_m\setminus L_n^*$ with witnesses $U_{\{0\},m}(a,b,c)$
 and $U_n(b,a,c)$.

\begin{proposition}
\label{prop:bound1}
Let $K_m$ and $L_n$ be two regular languages with complexities $m$ and $n$.
Then the complexities of $K_m\circ L_n^*$ and $L_n^*\setminus K_m$ are at most $m(2^{n-1}+2^{n-2}-1)+1$, for $n\ge 3$.
\end{proposition}

\begin{proof}
Let $\cD_1=(Q_1,\Sig,\delta_1,0,F_1)$ with
$Q_1=\{0,\ldots,m-1\}$ be the DFA of $K_m$, and 
let $\cD_2=(Q_2,\Sig,\delta_2,0,F_2)$ with $Q_2=\{0,\ldots,n-1\}$ be the DFA of $L_n$.
Construct $\cN_2$, an NFA accepting $L_n^*$, by adding a new final state $s$ to $\cD_2$, with the same outgoing transitions as state $0$, and $\eps$-transitions from each final state in $F_2$ to $0$. Now $\cN_2$ has initial state $\{s\}$ instead of $\{0\}$.
See Fig.~\ref{fig:KULstar} for an illlustration.
Let $\cS_2$
be the minimal DFA obtained from $\cN_2$ by the subset construction and minimization,
and let $\cP$ be the direct product of $\cD_1$ and $\cS_2$.

\begin{figure}[hbt]
\begin{center}
\input KULstar.eepic
\end{center}
\caption{DFA $\cD_1$ of $U_4(a,b,c)$ and NFA $\cN_2$ of $(U_5(b,a,c))^*$.} 
\label{fig:KULstar}
\end{figure}

For all five boolean operations, the states of $\cP$ are ordered pairs, where the first element is a state $i \in Q_1$ and the second is either $\{s\}$ or a subset of $Q_2$.
Because of the $\eps$-transitions, the allowable states are $(0, \{s\})$, 
all states of the form $(i,S)$ where $S$ is non-empty and $S\cap F_2=\emp$, 
and all states of the form $(i,S)$ where  $S$ contains at least one final state together with 0. The total number of possible states is largest if there is only one final state, say $n-1$.
Hence the number of states in $\cP$ cannot exceed 1 plus
 $m(2^{n-1}-1)$ for states of the form $(i,S)$ where $S$ is non-empty and $n-1\notin S$, 
and  $m2^{n-2}$ for states of the form $(i,S)$ where  
$0, n-1 \in S$.
Therefore the complexity of $K_m \circ L_n^*$ and $L_n^*\setminus K_m$ cannot exceed 
$1+ m(2^{n-1} + 2^{n-2} - 1)$.
\qed
\end{proof}

\begin{theorem}[$K\circ L^*$]
\label{thm:KULstar}
Let $K_m=U_m(a,b,c)$ and $L_n=U_n(b,a,c)$. For $m,n\ge 3$,
the complexities of  $K_m\cup L_n^*$, $K_m\oplus L_n^*$, and $L_n^*\setminus K_m$ are all $m(2^{n-1}+2^{n-2}-1)+1$.
Let $K'_m$ be the language  $U_{\{0\},m}(a,b,c)$. 
Then the complexities of  $K_m'\cap L_n^*$ and $K_m'\setminus L_n^*$ are also $m(2^{n-1}+2^{n-2}-1)+1$.
\end{theorem}
\begin{proof}
Let the various automata be defined as in the proof of Proposition~\ref{prop:bound1}, but this time with $K_m=U_m(a,b,c)$ and $L_n=U_n(b,a,c)$.
We show that all $m(2^{n-1} + 2^{n-2} - 1) + 1$ allowable states of $\cP$ are reachable.
We use the notation $(i, S) \der{w} (j, T)$ to denote that state $(j,T)$ is reached from $(i,S)$ by  word $w$.
We have $(0, \{s\}) \der{c} (0, \{0\}) \der{(ba)^{i-1}} (i, \{0\})$ for $2 \le i \le m-1$. 
If $m$ is odd, $(0, \{0\}) \der{a^{m+1}} (1, \{0\})$; if $m$ is even, $(0, \{0\}) \der{a^{m-1}ca} (1, \{0\})$.

Brzozowski showed in \cite{Brz12} that all allowable states of $\cN_2$ are reachable from $\{0\}$ by words in $\{a,b\}^*$.
These words act as permutations on $\cD_1$.
To reach state $(i,S)$ apply the word $w$ that takes $\{0\}$ to $S$ in $\cN_2$ to state 
$(j,\{0\})$, where $j$ is such that $j \der{w} i$.
Therefore all the allowable states are reachable.

For distinguishability, first consider two states $(i, S)$ and $(j, T)$, where $S \neq T$.
Then there is a $k$ either in $S \bs T$ or in  $T \bs S$; without loss of generality, assume $k\in S \bs T$.
By applying $b^{n-1-k}$, we reach states $(i', S')$ and $(j', T')$, where $n-1\in S'\bs T'$.
Note that applying some cyclic shift $a^l$ to $\cD_1$, we reach states $(i'', S'')$ and $(j'', T'')$, where $n-1\in S''\bs T''$.
These states are distinguishable for the boolean operations as follows:
\bi
\item $K_m \cup L_n^*, K_m \oplus L_n^*, L_n^* \bs K_m$: apply a cyclic shift so $i',j'$ are non-final in $\cD_1$. This is possible since as $\cD_1$ has a single final state and $m\ge 3$.
\item $K'_m \cap L_n^*$: map $i$ to the final state of $\cD_1$.
\item $K'_m \bs L_n^*$: map $j$ to the final state of $\cD_1$.
\ei

Now consider two states $(i, S)$ and $(j, S)$, $i < j$. We may assume $j < m-1$ because,
 since $m \ge3$, we can apply a cyclic shift of $a$'s so that neither $i$ nor $j$ is equal to $m-1$. Doing so might change $S$ to $S'$, but  $S'$ is the same in both states and $S'$ remains non-empty. 
The states are distinguishable as follows:
\bi
\item $K_m \cup L_n^*, K_m \oplus L_n^*, K'_m \bs L_n^*$: apply $c$ so that $n-1 \notin S$, then $a^k$ for some $k$ to map $j$ to a final state.
\item $K'_m \cap L_n^*, L_n^* \bs K_m$: since $S$ is non-empty, apply a cyclic shift so $n-1 \in S$, then another shift so $j$ is final, and hence $i$ is non-final.
\ei

Finally, note that only states $(0, \{s\})$ and $(0, \{0\})$ reach $(1, \{1\})$ on applying $a$; therefore by the previous argument, $(0, \{s\})$ is distinguishable from all other states except possibly $(0, \{0\})$.
Note now that states $(0,\{s\})$ and $(0,\{0\})$ are distinguishable in 
$K_m\cup L_n^*$, $K_m\oplus L_n^*$ and $L_n^*\setminus K_m$, but
equivalent in $K_m\cap L_n^*$ and $K_m\setminus L_n^*$.
Hence we cannot have the same witnesses for both intersection and union.
However, the choice of final states distinguishes $(0, \{s\})$ from $(0, \{0\})$ for $K'_m \cap L_n^*$ and $K'_m \bs L_n^*$. 
Therefore all  reachable states are distinguishable.
\qed
\end{proof}

\section{Boolean Operations with Two Starred Arguments}
\label{sec:tworev}

 Gao, Kari and Yu~\cite{GKY_TCS12} showed that the bounds for  $K_m^*\cup L_n^*$ and 
$K_m^*\cap L_n^*$ are both $(2^{m-1}+2^{m-2}-1)(2^{n-1}+2^{n-2}-1)+1$.
They used the following  DFA's over alphabet $\Sig=\{a,b,c,d\}$:
For $K$, let
$\cD_K=(Q_K,\Sig,\delta_K,0,\{m-1\})$, with $Q_K=\{0,\ldots,m-1\}$,
$a:(0,\ldots,{m-1})$,
$b$ defined by $\delta_K(0,b)=0$, $\delta_K(i,b)=i+1 \pmod m $, for $i=1,\ldots,m-1$, and 
$c,d:\one_{Q_K}$.
For $L$, let $\cD_L=(Q_L,\Sig,\delta_L,0,\{n-1\})$, with $Q_L=\{0,\ldots,n-1\}$, 
$a,b:\one_{Q_L}$, 
$c:(0,\ldots n-1)$, and 
$d$ defined by $\delta_K(0,d)=0$, $\delta_K(i,d)=i+1 \pmod n $, for $i=1,\ldots,n-1$.

We extend these results  to $K_m\oplus L_n^*$ and $K_n^*\setminus L_m^*$, for which we now derive upper bounds.

\begin{proposition}
\label{prop:bound2}
Let $K_m$ and $L_n$ be two regular languages with complexities $m$ and $n$.
Then the complexities of $K_m^*\circ L_n^*$ are at most $(2^{m-1}+2^{m-2}-1)(2^{n-1}+2^{n-2}-1)+1$ for $m,n\ge 3$.
\end{proposition}
\begin{proof}
Let $\cD_1=(Q_1,\Sig,\delta_1,0,F_1)$  be the DFA of $K_m$, and 
$\cD_2=(Q_2,\Sig,\delta_2,0,F_2)$,   the DFA of $L_n$.
Let $\cN_1$ ($\cN_2$) be the NFA  for $K_m^*$ ($L_n^*$) obtained by adding a new initial and final state $s_1$ ($s_2$), transitions from state $s_1$ ($s_2$) the same as from $0$ in $\cD_1$ ($\cD_2$), and  an $\eps$-transition from each final state of $\cD_1$ ($\cD_2$) to the initial state 0 of $\cD_1$ ($\cD_2$).
See Fig.~\ref{fig:KstarULstar} for an example of this construction.
Let 
$\cS_1$ and $\cS_2$ be the minimal DFA's 
obtained from $\cN_1$ and $\cN_2$ by the subset construction and minimization.
Finally, let $\cP$ be the direct product of $\cS_1$ and $\cS_2$.

\begin{figure}[t]
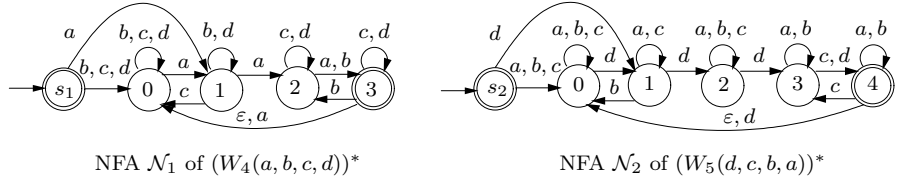

\begin{center}
\input KstarULstar.eepic
\end{center}
\caption{NFA's $\cN_1$ and $\cN_2$ of $(W_4(a,b,c,d))^*$ and $(W_5(d,c,b,a))^*$.} 
\label{fig:KstarULstar}
\end{figure}

The states of $\cP$ are ordered pairs, where the first element is a
subset of $\{s_1\}\cup Q_1$  and the second is a subset of $\{s_2\}\cup Q_2$.
Note that $s_1$ and $s_2$ can only appear in the initial state $(\{s_1\},\{s_2\})$ of $\cP$.
After any input  is applied to $\cP$, the state has the form $(S,T)$, where
$S$ is a state of $\cS_1$ other than $\{s_1\}$ (there are at most $2^{m-1}+2^{m-2}-1$ such states),
and $T$ is a state of $\cS_2$ other than $\{s_2\}$ (there are at most $2^{n-1}+2^{n-2}-1$ such states), and this is independent of the witnesses used.
Thus $(2^{m-1}+2^{m-2}-1)(2^{n-1}+2^{n-2}-1)+1$ is an upper bound for the number of states of the DFA for $K^*\circ L^*$.
\qed
\end{proof}

\begin{theorem}[$K^*\circ L^*$]
Let $K_m=W_m(a,b,c,d)$ and $L_n=W_n(d,c,b,a)$. 
For $m,n\ge 3$, 
the complexities of $K_m^*\cup L_n^*$ and $K_m^*\cap L_n^*$
are $(2^{m-1}+2^{m-2}-1)(2^{n-1}+2^{n-2}-1)+1$.
If  $K'_m$ is the language 
of $\cW_{\{0\},m}$, then the complexities of 
$(K'_m)^*\setminus L_n^*$ and $(K'_m)^*\oplus L_n^*$
are also $(2^{m-1}+2^{m-2}-1)(2^{n-1}+2^{n-2}-1)+1$.
\end{theorem}
\begin{proof}
Let the various automata be defined as in the proof of Proposition~\ref{prop:bound2}, but this time with $K_m=W_m(a,b,c,d)$ and $L_n=W_n(d,c,b,a)$.
We now show that all $(2^{m-1}+2^{m-2}-1)(2^{n-1}+2^{n-2}-1)+1$ allowable states discussed in Proposition~\ref{prop:bound2} are reachable.

We first show that all allowable subsets of $Q_1$ are reachable in $\cD_1$, ignoring $\cD_2$.
First, $\{s_1\} \der{c} \{0\} \der{a^{m-1}} \{0,m-1\}$.
Suppose all states $S$ with $\{0,m-1\} \subseteq S \subseteq Q_1$, $|S| = k$, $k \ge 2$ are reachable.
 All states $S$ with
$  \{0,1\} \subseteq S\subseteq Q_1$ of size $k$ are now reachable by applying $a$.
If $S = \{i_1, \dots, i_k\}$ with $i_1 < \cdots < i_k < m-1$,
let $j = i_2 - i_1 - 1$;
then $\{0, 1, i_3 - j - i_1, \dots, i_k - j - i_1\} \der{(ac)^ja^{i_1}} S$.

Now states $\{0,m-1\}  \subseteq S$ of size $k+1$ can now be reached as follows: $\{i_1-1,  \dots, i_{k-1}-1, m-2\} \der{a} \{0,i_1,  \dots, i_{k-1}, m-1\}$.

Therefore all allowable states of $\cD_1$ are reachable by words in $\{a,c\}^*$.

In $\cN_2$, $a$ and $c$ map states $s_2$ and 0 to 0. 
Therefore all allowable states of $\cP$ of the form $(S, \{0\})$ are reachable.
A symmetric argument shows that all states $T$ of $\cD_2$ are reachable by words in $\{b^2, d\}^*$ (as $b^2$ and $b$ are the same transformation on $\cD_2$). 
All of these words map states $S \subseteq Q_1$ to themselves, except in the case $0, m-1 \notin S$, $m-2 \in S$.
Let $S = \{i_1, \dots, i_k\}$ be such a state; then for all allowable $T$, $(\{i_1 - 1, \dots, i_k - 1\}, T)$ is reachable, and reaches $(S,T)$ when $a$ is applied.
Therefore all allowable states are reachable.

Next we show that all the states of $\cP$ are distinguishable.
Recall that for $K_m^* \cup L_n^*$ and $K_m^* \cap L_n^*$, we use $\{m-1\}$ as the final state of $\cN_1$,
and for $(K'_m)^* \oplus L_n^*$ and $(K'_m)^* \bs L_n^*$, we use $\{0\}$.

Suppose we have states $(S_1,T_1)$, $(S_2, T_2)$ with $T_1 \neq T_2$.
Then there is a $k$ either in $T_1 \bs T_2$ or in  $T_2 \bs T_1$; without loss of generality, assume $k\in T_1 \bs T_2$.
By applying $d^{n-1-k}$, we reach states $(S_1, T_1')$ and $(S_2, T_2')$, where $n-1\in T_1'\bs T_2'$.
Apply $c^2ac^2$ so that $T_1'$ and $T_2'$ are unchanged, but now $1,2 \notin S'_1 \cup S'_2$.
Then apply $a^{m-2}$ so $0, m-1 \notin S''_1 \cup S''_2$.
This distinguishes the two states for $K_m^* \cup L_n^*$ and $(K'_m)^* \oplus L_n^*$.
For $K_m^* \cap L_n^*$, since $S_1 \neq \emp$, we may apply a cyclic shift to $\cD_1$ so that $m-1 \in S'_1$ to distinguish the states.
For $(K'_m)^* \bs L_n^*$, we can assume that $h \in S''_2$, and  use 
$a^{m-1-h}$ to map $S''_2$ to $S'''_2$, where  $\{0,m-1\}
\subseteq S'''_2$. This also maps $S''_1$ to $S'''_1$, and keeps $T'_1$ and
$T'_2$ unchanged. Since $n-1 \in T'_1 \bs T'_2$, we have $(S'''_1, T'_1)$ is
non-final and $(S'''_2, T'_2)$ is final for $(K'_m)^* \bs L_n^*$.

Now suppose $S_1 \neq S_2$. 
For $K_m^* \cup L_n^*$ and $K_m^* \cap L_n^*$ the above argument is symmetric.
For the other two operations, apply a cyclic shift so that $m-1 \in S'_1 \bs S'_2$.
Now apply $(cba)^{m-3}$ so that $m-1 \in S''_1 \bs S''_2$, and $2,\dots, m-2 \notin S''_1 \cup S''_2$.
Apply $a$ so that $0 \in S'''_1 \bs S'''_2$.
Then as above, apply $b^2d^{n-2}$ so that $n-1 \notin T'_1 \cup T'_2$, while leaving $S'''_1$ and $S'''_2$ unchanged.
This distinguishes the states for $(K'_m)^* \bs L_n^*$ and $(K'_m)^* \oplus L_n^*$.

Therefore all $(2^{m-1} + 2^{m-2} - 1)(2^{n-1} + 2^{n-2} - 1)$ states of the form $(S,T)$ are distinguishable.
It remains to distinguish $(\{s_1\},\{s_2\})$ from the other states.
As in Theorem \ref{thm:KULstar}, $(\{s_1\}, \{s_2\})$ is distinguished from all states except $(\{0\},\{0\})$ by  $a$. 
It is distinguishable from $(\{0\}, \{0\})$ by the choice of final state of $\cD_1$.
\qed 
\end{proof}

\section{Products with Starred Arguments}
\label{sec:prodstar}

\subsection{The Language $KL^*$}

The complexity of $KL^*$ was studied by Cui, Gao, Kari, and Yu~\cite{CGKY_IJFCS12}.
They showed that $m(2^{n-1}+ 2^{n-2})-2^{n-2}$ is a tight bound using the following witnesses 
over alphabet $\Sig=\{a,b,c\}$:
For $K$, let
$\cD_K=(Q_K,\Sig,\delta_K,q_0,\{m-1\})$, with $Q_K=\{q_0,\ldots,q_{m-1}\}$,
$a:(q_0,\ldots,q_{m-1})$,
$\delta_K(q_i,b)=q_{i+1}$ for $i=0,\ldots,m-3$, $\delta_K(q_{m-2},b)=q_0$,  $\delta_K(q_{m-1},b)=q_{m-2}$, and
$\delta_K(q_i,c)=q_{i+1}$ for $i=0,\ldots,m-3$, $\delta_K(q_{m-2},c)=q_0$, $\delta_K(q_{m-1},c)=q_{m-1}$.
For $L$, let $\cD_L=(Q_L,\Sig,\delta_L,0,\{n-1\})$, with $Q_L=\{0,\ldots,n-1\}$, 
$a:(0,\ldots n-1)$,
$\delta_L(0,b)=0$, $\delta_L(i,b)=i+1$ for $i=1,\ldots,n-2$, $\delta(n-1,b)=1$;
$c:{n-1\choose 1}$.
We  prove  that two permutationally equivalent dialects of $U_n(a,b,c)$ also meet the bound.

\begin{figure}[b]
\begin{center}
\setlength{\unitlength}{0.00039370in}
\begingroup\makeatletter\ifx\SetFigFont\undefined%
\gdef\SetFigFont#1#2#3#4#5{%
  \reset@font\fontsize{#1}{#2pt}%
  \fontfamily{#3}\fontseries{#4}\fontshape{#5}%
  \selectfont}%
\fi\endgroup%
{\renewcommand{\dashlinestretch}{30}
\begin{picture}(11440,1748)(0,-10)
\put(7077,1069){\makebox(0,0)[lb]{\smash{{\SetFigFont{9}{10.8}{\familydefault}{\mddefault}{\updefault}$a,b$}}}}
\put(10002.500,1135.929){\arc{394.717}{2.4948}{6.9299}}
\blacken\thicklines
\path(10162.033,1157.097)(10160.000,1017.000)(10233.998,1135.977)(10162.033,1157.097)
\thinlines
\put(11127.500,1120.929){\arc{394.717}{2.4948}{6.9299}}
\blacken\thicklines
\path(11287.033,1142.097)(11285.000,1002.000)(11358.998,1120.977)(11287.033,1142.097)
\thinlines
\put(3041.500,1135.929){\arc{394.717}{2.4948}{6.9299}}
\blacken\thicklines
\path(3201.033,1157.097)(3199.000,1017.000)(3272.998,1135.977)(3201.033,1157.097)
\thinlines
\put(4181.500,1113.929){\arc{394.717}{2.4948}{6.9299}}
\blacken\thicklines
\path(4341.033,1135.097)(4339.000,995.000)(4412.998,1113.977)(4341.033,1135.097)
\thinlines
\put(806.500,1098.929){\arc{394.717}{2.4948}{6.9299}}
\blacken\thicklines
\path(966.033,1120.097)(964.000,980.000)(1037.998,1098.977)(966.033,1120.097)
\thinlines
\put(6739.500,1112.929){\arc{394.717}{2.4948}{6.9299}}
\blacken\thicklines
\path(6899.033,1134.097)(6897.000,994.000)(6970.998,1112.977)(6899.033,1134.097)
\thinlines
\put(8882,710){\ellipse{630}{630}}
\put(10003,717){\ellipse{630}{630}}
\put(6732,702){\ellipse{630}{630}}
\put(7789,720){\ellipse{630}{630}}
\put(11117,709){\ellipse{630}{630}}
\put(11118,707){\ellipse{540}{540}}
\put(5545,694){\ellipse{630}{630}}
\put(5546,696){\ellipse{540}{540}}
\put(817,678){\ellipse{630}{630}}
\put(4164,697){\ellipse{630}{630}}
\put(3035,709){\ellipse{630}{630}}
\put(1900,673){\ellipse{630}{630}}
\path(7573,500)(6988,500)
\blacken\thicklines
\path(7123.000,537.500)(6988.000,500.000)(7123.000,462.500)(7123.000,537.500)
\thinlines
\path(8051,935)(8636,935)
\blacken\thicklines
\path(8501.000,897.500)(8636.000,935.000)(8501.000,972.500)(8501.000,897.500)
\thinlines
\path(9147,927)(9732,927)
\blacken\thicklines
\path(9597.000,889.500)(9732.000,927.000)(9597.000,964.500)(9597.000,889.500)
\thinlines
\path(10272,927)(10857,927)
\blacken\thicklines
\path(10722.000,889.500)(10857.000,927.000)(10722.000,964.500)(10722.000,889.500)
\thinlines
\path(6957,942)(7542,942)
\blacken\thicklines
\path(7407.000,904.500)(7542.000,942.000)(7407.000,979.500)(7407.000,904.500)
\thinlines
\path(4490,717)(5202,717)
\blacken\thicklines
\path(5067.000,679.500)(5202.000,717.000)(5067.000,754.500)(5067.000,679.500)
\thinlines
\path(11127,394)(11127,34)(6717,34)(6717,349)
\blacken\thicklines
\path(6754.500,214.000)(6717.000,349.000)(6679.500,214.000)(6754.500,214.000)
\thinlines
\path(2142,897)(2772,897)
\blacken\thicklines
\path(2637.000,859.500)(2772.000,897.000)(2637.000,934.500)(2637.000,859.500)
\thinlines
\path(3298,912)(3928,912)
\blacken\thicklines
\path(3793.000,874.500)(3928.000,912.000)(3793.000,949.500)(3793.000,874.500)
\thinlines
\path(1017,904)(1647,904)
\blacken\thicklines
\path(1512.000,866.500)(1647.000,904.000)(1512.000,941.500)(1512.000,866.500)
\thinlines
\path(1655,492)(1070,492)
\blacken\thicklines
\path(1205.000,529.500)(1070.000,492.000)(1205.000,454.500)(1205.000,529.500)
\thinlines
\path(12,687)(499,687)
\blacken\thicklines
\path(379.000,657.000)(499.000,687.000)(379.000,717.000)(379.000,657.000)
\thinlines
\path(4175,372)(4175,12)(754,20)(754,335)
\blacken\thicklines
\path(791.500,200.000)(754.000,335.000)(716.500,200.000)(791.500,200.000)
\thinlines
\path(5884,709)(6394,709)
\blacken\thicklines
\path(6259.000,671.500)(6394.000,709.000)(6259.000,746.500)(6259.000,671.500)
\thinlines
\path(5547,1024)(5547,1699)(7797,1699)(7797,1069)
\blacken\thicklines
\path(7759.500,1204.000)(7797.000,1069.000)(7834.500,1204.000)(7759.500,1204.000)
\put(6649,630){\makebox(0,0)[lb]{\smash{{\SetFigFont{9}{10.8}{\rmdefault}{\mddefault}{\updefault}$0$}}}}
\put(7707,637){\makebox(0,0)[lb]{\smash{{\SetFigFont{9}{10.8}{\rmdefault}{\mddefault}{\updefault}$1$}}}}
\put(8816,630){\makebox(0,0)[lb]{\smash{{\SetFigFont{9}{10.8}{\rmdefault}{\mddefault}{\updefault}$2$}}}}
\put(9919,638){\makebox(0,0)[lb]{\smash{{\SetFigFont{9}{10.8}{\rmdefault}{\mddefault}{\updefault}$3$}}}}
\put(11036,638){\makebox(0,0)[lb]{\smash{{\SetFigFont{9}{10.8}{\rmdefault}{\mddefault}{\updefault}$4$}}}}
\put(8210,1064){\makebox(0,0)[lb]{\smash{{\SetFigFont{9}{10.8}{\familydefault}{\mddefault}{\updefault}$b$}}}}
\put(9333,1063){\makebox(0,0)[lb]{\smash{{\SetFigFont{9}{10.8}{\familydefault}{\mddefault}{\updefault}$b$}}}}
\put(10452,1063){\makebox(0,0)[lb]{\smash{{\SetFigFont{9}{10.8}{\familydefault}{\mddefault}{\updefault}$b$}}}}
\put(5456,607){\makebox(0,0)[lb]{\smash{{\SetFigFont{9}{10.8}{\rmdefault}{\mddefault}{\updefault}$s$}}}}
\put(7057,215){\makebox(0,0)[lb]{\smash{{\SetFigFont{9}{10.8}{\familydefault}{\mddefault}{\updefault}$a,c$}}}}
\put(1790,634){\makebox(0,0)[lb]{\smash{{\SetFigFont{9}{10.8}{\rmdefault}{\mddefault}{\updefault}$q_1$}}}}
\put(2922,648){\makebox(0,0)[lb]{\smash{{\SetFigFont{9}{10.8}{\rmdefault}{\mddefault}{\updefault}$q_2$}}}}
\put(4084,625){\makebox(0,0)[lb]{\smash{{\SetFigFont{9}{10.8}{\rmdefault}{\mddefault}{\updefault}$q_3$}}}}
\put(1116,1016){\makebox(0,0)[lb]{\smash{{\SetFigFont{9}{10.8}{\familydefault}{\mddefault}{\updefault}$a,b$}}}}
\put(716,620){\makebox(0,0)[lb]{\smash{{\SetFigFont{9}{10.8}{\rmdefault}{\mddefault}{\updefault}$q_0$}}}}
\put(5887,866){\makebox(0,0)[lb]{\smash{{\SetFigFont{9}{10.8}{\familydefault}{\mddefault}{\updefault}$c$}}}}
\put(4805,897){\makebox(0,0)[lb]{\smash{{\SetFigFont{9}{10.8}{\familydefault}{\mddefault}{\updefault}$\eps$}}}}
\put(3530,1046){\makebox(0,0)[lb]{\smash{{\SetFigFont{9}{10.8}{\familydefault}{\mddefault}{\updefault}$a$}}}}
\put(1146,169){\makebox(0,0)[lb]{\smash{{\SetFigFont{9}{10.8}{\familydefault}{\mddefault}{\updefault}$b,c$}}}}
\put(2826,1437){\makebox(0,0)[lb]{\smash{{\SetFigFont{9}{10.8}{\familydefault}{\mddefault}{\updefault}$b,c$}}}}
\put(3958,1429){\makebox(0,0)[lb]{\smash{{\SetFigFont{9}{10.8}{\familydefault}{\mddefault}{\updefault}$b,c$}}}}
\put(741,1414){\makebox(0,0)[lb]{\smash{{\SetFigFont{9}{10.8}{\familydefault}{\mddefault}{\updefault}$c$}}}}
\put(2315,1039){\makebox(0,0)[lb]{\smash{{\SetFigFont{9}{10.8}{\familydefault}{\mddefault}{\updefault}$a$}}}}
\put(9762,1459){\makebox(0,0)[lb]{\smash{{\SetFigFont{9}{10.8}{\familydefault}{\mddefault}{\updefault}$a,c$}}}}
\put(8645,1451){\makebox(0,0)[lb]{\smash{{\SetFigFont{9}{10.8}{\familydefault}{\mddefault}{\updefault}$a,c$}}}}
\put(6584,1400){\makebox(0,0)[lb]{\smash{{\SetFigFont{9}{10.8}{\familydefault}{\mddefault}{\updefault}$c$}}}}
\put(9208,155){\makebox(0,0)[lb]{\smash{{\SetFigFont{9}{10.8}{\familydefault}{\mddefault}{\updefault}$b,\eps$}}}}
\put(10840,1478){\makebox(0,0)[lb]{\smash{{\SetFigFont{9}{10.8}{\familydefault}{\mddefault}{\updefault}$a,c$}}}}
\put(2338,147){\makebox(0,0)[lb]{\smash{{\SetFigFont{9}{10.8}{\familydefault}{\mddefault}{\updefault}$a$}}}}
\put(5682,1384){\makebox(0,0)[lb]{\smash{{\SetFigFont{9}{10.8}{\familydefault}{\mddefault}{\updefault}$a,b$}}}}
\thinlines
\put(8884.500,1120.929){\arc{394.717}{2.4948}{6.9299}}
\blacken\thicklines
\path(9044.033,1142.097)(9042.000,1002.000)(9115.998,1120.977)(9044.033,1142.097)
\end{picture}
}
\end{center}
\caption{Witness $\cN$ for $\cT_4(a,b,c)(\cT_5(b,a,c))^*$.} 
\label{fig:KLstar}
\end{figure}
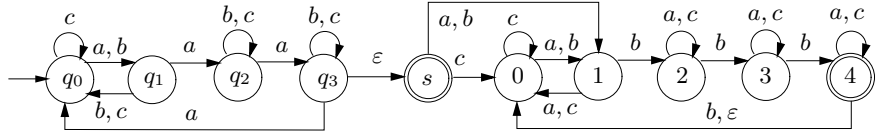

\begin{theorem}[$KL^*$]
Let $K_m=T_m(a,b,c)$, and $L_n=T_n(b,a,c)$.
For $m, n\ge 3$, the complexity of $K_mL_n^*$ is  $m(2^{n-1}+2^{n-2})-2^{n-2}$.
\end{theorem}

\begin{proof}
Let $\cD_1=(Q_1,\Sig,\delta_1,q_0,\{q_{m-1}\})$ with $Q_1=\{q_0,\ldots,q_{m-1}\}$ be the DFA of $K_m$, and let
$\cD_2=(Q_2,\Sig,\delta_2,0,\{n-1\})$ with $Q_2=\{0,\ldots,{n-1}\}$ be the DFA of $L_n$.
Let $\cN_2$ be the NFA for $L_n^*$, and let $\cN$ be the NFA for the product $K_mL_n^*$.
Figure~\ref{fig:KLstar} shows our witnesses $\cT_4(a,b,c)$ and $\cT_5(b,a,c)$ and the NFA $\cN$ for $KL^*$.
We perform the subset construction and minimization of $\cN$ to obtain the DFA $\cP$ for the product $KL^*$. 
 
The states of $\cP$ are subsets of $Q_1\cup Q_2\cup \{s\}$.
Note that $q_{m-1}$ cannot appear in a state of  $\cP$ without $s$, and 
\emph{vice versa}.
Also, $n-1$ cannot appear without 0, but 0 can appear without $n-1$.
Each state of $\cD$ must contain exactly one of $\{q_0\},\ldots,\{q_{m-2}\}$ or
$\{q_{m-1},s\}$,
and either a (possibly empty) subset of $Q_2$ not containing $n-1$, or subset of $Q_2$ containing both  $n-1$ and $0$.
Hence there are at most $m(2^{n-1}+2^{n-2})$ reachable subsets; we now show that all these subsets can be reached.

Set $\{q_0\}$ is the initial state of $\cP$,
set $\{q_i\}$ for $i \le m-2$ is reached by $a^i$, and $\{q_{m-1},s\}$, by $a^{m-1}$.

Suppose all allowable states of the form $\{q_{m-1},s\} \cup S$, $|S| \le k$, $k \ge 0$, are reachable.
Let $S \subseteq Q_2$, $|S| = k+1$.
If $1 \in S$ and $0 \notin S$, then we have $\{q_{m-1},s\} \cup (S \bs \{1\}) \der{a} \{q_0\} \cup S$.
If $0,1 \in S$, then $\{q_{m-1},s\} \cup (S \bs \{0\}) \der{a} \{q_0\} \cup S$.
If $0 \in S$ and $1 \notin S$, then $\{q_{m-1},s\} \cup (S \bs \{0\}) \der{ac} \{q_0\} \cup S$.
Therefore all states $\{q_0\} \cup S$, $|S| = k+1$, and either $0 \in S$ or $1 \in S$, are reachable.
Every state $\{q_0\} \cup S$, where $n-1 \notin S$, is reachable by an even number of $b$'s from a state containing either 0 or 1.
Every $S = \{0, i_1, \dots, i_{k-1}, n-1\}$ is also reachable in this way (by mapping either 0 or 1 to $i_1$).
So all states $\{q_0\} \cup S$, $|S| = k+1$, are reachable.
By applying cyclic shifts $a^i$, all states $\{q_i\} \cup S$, $i < m-1$ and $\{q_{m-1}, s\} \cup S$ are reachable.

Any state of the form $\{q_{m-1},s\}\cup T$, where $T\subset Q_L \bs \{0,n-1\}$, is equivalent
to $\{q_{m-1},s,0\}\cup T$, as they are both final and are mapped to the same state under any input. 
So the number of distinguishable states of $\cD$ is at most
$m(2^{n-1}+2^{n-2})-2^{n-2}$.
We  prove that there are precisely that many distinguishable states.

Consider two states of the form $\{q_i\} \cup S$, $\{q_{m-1},s\} \cup T$, where $i < m-1$. 
These states are distinguished by $cb^{n-2}$.
Any pair $\{q_i\} \cup S$, $\{q_j\} \cup T$, $i \neq j$ can by transformed into states of this form by applying a cyclic shift.
Now consider $\{q_i\} \cup S$, $\{q_i\} \cup T$, $S \neq T$, $i < m-1$.
There exists a cyclic shift $b^k$ which transforms the states so that $n-1 \in S \oplus T$, and this distinguishes the states.

Then the only remaining case is $\{q_{m-1},s\} \cup S$, $\{q_{m-1},s\} \cup T$, and $S \neq T$.
As we stated earlier, if $S \oplus T = \{0\}$ then the states are indistinguishable.
Otherwise, let $k \in S \oplus T$, $k > 0$.
Apply $b^{n-1-k}$ so that $n-1 \in S \oplus T$.
Then applying $a$ to map $\{q_{m-1},s\}$ to $\{q_0, 1\}$ distinguishes the states.
\qed
\end{proof}

\subsection{The Language $K^*L$}

Cui, Gao, Kari and Yu~\cite{CGKY_TCS12b} proved using quaternary witnesses that the complexity of $K^*L$ is
$5\cdot 2^{m+n-3}-2^{m-1}-2^n+1$. Let $\Sig=\{a,b,c,d\}$.
For $K$ they used
$\cD_K=(Q_K,\Sig,\delta_K,q_0,\{m-1\})$, with $Q_K=\{q_0,\ldots,q_{m-1}\}$,
$a:(q_0,\ldots,q_{m-1})$,
$\delta_K(q_0,b)=q_0$, $\delta_K(q_{i},b)=i+1 \text{ mod } m$  for $i=1,\ldots,m-1$, 
 and
$c,d:\one_{Q_K}$.
For $L$, their witness was  $\cD_L=(Q_L,\Sig,\delta_L,0,\{n-1\})$, with $Q_L=\{0,\ldots,n-1\}$, 
$a,b:\one_{Q_L}$,
$c:(0,\ldots n-1)$,
$d:{Q_L \choose 0}$.
We show here that two quaternary permutationally equivalent languages also work.

\begin{figure}[h]
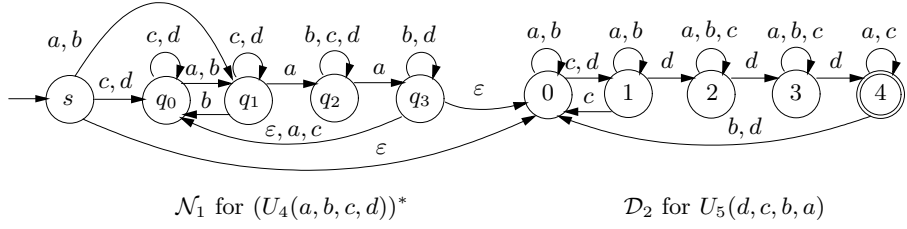

\begin{center}
\input KstarL.eepic
\end{center}
\caption{NFA $\cN$ for $(U_4(a,b,c,d))^*\;U_5(d,c,b,a)$.} 
\label{fig:KstarL}
\end{figure}

\begin{theorem}[$K^*L$]
\label{thm:KStarL}
Let $K_m = U_m(a,b,c,d)$ and $L_n = U_n(d,c,b,a)$. For $m,n \ge 3$, the complexity of $K_m^*L_n$ is $5 \cdot 2^{m+n-3} - 2^{m-1} - 2^n + 1$.
\end{theorem}
\begin{proof}
Let $\cD_1=(Q_1,\Sig,\delta_1,q_0,\{q_{m-1}\})$ with $Q_1=\{q_0,\ldots,q_{m-1}\}$ be the DFA of $K_m$, and let
$\cD_2=(Q_2,\Sig,\delta_2,0,\{n-1\})$ with $Q_2=\{0,\ldots,{n-1}\}$ be the DFA of $L_n$.
Let $\cN_1$ be the NFA for $K_m^*$, and let $\cN$ be the NFA for the product $K_m^*L_n$.
We perform the subset construction and minimization of $\cN$ to obtain the DFA $\cP$ for the product $K^*L$. 
The construction is illustrated in Fig.~\ref{fig:KstarL}.

Owing to the $\eps$-transitions, the allowable states of the DFA are $\{s,0\}$, all 
$(2^{m-1}-1)(2^n-1)$ subsets of the form $S \cup T$ where $\emp \subsetneq S \subseteq Q_1$, , $q_{m-1}\notin S$,
$\emp \subsetneq T \subseteq Q_2$, 
and all $(2^{m-2}-1)(2^{n-1}-1)$ subsets of the form $S \cup T$, where 
 $q_{0},q_{m-1} \in S \subseteq Q_1$  and $0 \in T\subseteq Q_2$.
There are $5 \cdot 2^{m+n-3} - 2^{m-1} - 2^n + 2$ such subsets and we will now show that they are all reachable.

The initial state of $\cP$ is $\{s,0\}$.
It is known from \cite{Brz12} that all allowable subsets of $\cN_1$ are reachable by words in $\{a,b\}^*$. 
These inputs all map $0$ to itself, and hence all allowable states of the form $S \cup \{0\}$ are reachable.

If $q_{m-1} \notin S$ and $T = \{t_1, \dots t_k\}$, then $S \cup \{0, t_2 - t_1, \dots, t_k - t_1\} \der{d^{t_1}} S \cup T$.
Let $T = \{0, t_1, \dots, t_k\}$, $0 < t_1 < \cdots < t_k$, and 
 $S = \{q_{i_1}, \dots, q_{i_l}\}$, $i_1 < \cdots < i_l < m-1$.
Also, let $S' = \{q_{i_2 - i_1 - 1}, \dots, q_{i_l - i_1 -1}, q_{m-2}\}$ and $T' = \{t_1, \dots, t_k\}$.
Then $$S' \cup T' \der{ac^2} \{0, q_{i_2 - i_1}, \dots, q_{i_l - i_1}\} \cup T \der{a^{i_1}} S \cup T.$$
Moreover, $S \cup \{t_0, t_1 + t_0 , \dots, t_k + t_0\}$ can be reached from $S \cup T$ by $d^{t_0}$.
Combining these results shows that all allowable states $S \cup T$ with $q_{m-1} \notin S$ are reachable.
Finally, if $S = \{q_0, q_{i_1}, \dots, q_{i_k}, q_{m-1} \}$, and $0 \in T$, then $\{q_{i_1 - 1}, \dots, q_{i_l - 1}, q_{m-2}\} \cup T \der{a} S \cup T$.
Therefore all allowable states are reachable.

For distinguishability, first consider states $S_1 \cup T_1$, $S_2 \cup T_2$.
If $T_1 \neq T_2$, then applying a cyclic shift $d^k$ transforms the states so that $n-1 \in T_1 \oplus T_2$, distinguishing the states.
If $S_1 \neq S_2$, apply a cyclic shift $a^k$ so that $q_{m-1} \in S_1 \oplus S_2$.
Then apply $bd$ so that $0 \in T_1 \oplus T_2$, and the states are distinguishable by the previous case.

Finally, the initial state $\{s\} \cup \{0\}$ is indistinguishable from $\{q_0\} \cup \{0\}$, as any non-empty input transforms these two states into the same state.
So then there are $5 \cdot 2^{m+n-3} - 2^{m-1} - 2^n + 1$ distinguishable states.
\qed
\end{proof}

\subsection{The Language $K^*L^*$}
The combined operation $K^*L^*$ appears not to have been studied before. 
\begin{proposition}
\label{prop:KstarLstar}
The complexity of the operation $K_m^*L_n^*$
is at most 
$2^{m+n-1} - 2^{m-1} - 3 \cdot 2^{n-2} + 2$ for $m,n \ge 3$. 
\end{proposition}
\begin{proof}
Let $\cD_1=(Q_1,\Sig,\delta_1,q_0,F_1)$ with
$Q_1=\{q_0,\ldots,q_{m-1}\}$ be the DFA of $K_m$, and 
let $\cD_2=(Q_2,\Sig,\delta_2,0,F_2)$ with $Q_2=\{0,\ldots,n-1\}$ be the DFA of $L_n$.
Construct NFA's $\cN_1$ and $\cN_2$ accepting $K_m^*$ and $L_n^*$ by adding  new initial states $s_1$ and $s_2$, which are also final. 
Let $\cN$ be the NFA for $K_m^*L_n^*$, and 
let $\cP$ be the DFA  obtained by the subset construction and minimization of $\cN$. These constructions are illustrated in Fig.~\ref{fig:KstarLstar}.

\begin{figure}[hbt]
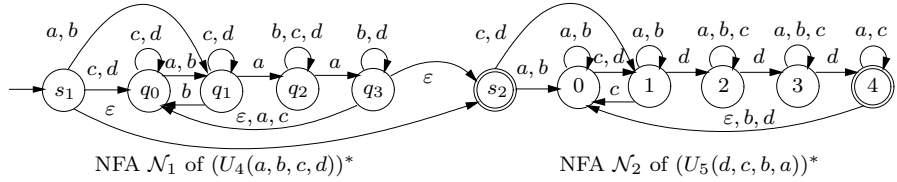

\begin{center}
\input KstarLstar.eepic
\end{center}
\caption{NFA $\cN$  of $(U_4(a,b,c,d))^*(U_5(d,c,b,a))^*$.} 
\label{fig:KstarLstar}
\end{figure}

The initial state of $\cP$ is $\{s_1,s_2\}$.
Note that any state $R$ of $\cP$ containing $s_2$ but not 0, is equivalent to $R\cup \{0\}$,
since both states are final because of $s_2$, and $s_2$ and 0 have identical outgoing transitions.
Hence we can ignore states like $R$ in our counting, and assume that every state containing $s_2$ also contains 0.
Due to the $\eps$-transitions, the allowable states of the DFA are $\{s_1,s_2\}$, and all subsets of the form $S \cup T$, where $\emp \subsetneq S \subseteq Q_1$, $\emp \subsetneq T \subseteq \{s_2\}\cup Q_2$, and fall into one of the following cases:
\bi
\item $S\cap F_1=\emp$, $T\cap F_2=\emp$;
\item $S\cap F_1=\emp$,  $T$ contains at least one state of $F_2$ and 0;
\item $S$ contains at least one state of $F_1$ and $s_2,0 \in T$.
\ei
One verifies that the possible number of states is greatest when there is only one final state, say $q_{m-1}$, in $F_1$ and only one final state, say $n-1$, in $F_2$. Hence we have the cases:
\bi
\item $q_{m-1} \notin S$, $n-1 \notin T$: $(2^{m-1}-1)(2^{n-1}-1)$ states;
\item $q_{m-1} \notin S$, $0, n-1 \in T$: $(2^{m-1}-1)2^{n-2}$ states;
\item $q_0, q_{m-1} \in S$, $s_2,0 \in T$: $2^{m+n-3}$ states.
\ei
Therefore there are a total of $2^{m+n-1} - 2^{m-1} - 3 \cdot 2^{n-2} + 2$ allowable states.
Hence the complexity of  $K_m^*L_n^*$
is at most $2^{m+n-1} - 2^{m-1} - 3 \cdot 2^{n-2} + 2$.\qed
\end{proof}

\begin{theorem}[$K^*L^*$]
Let $K_m = U_m(a,b,c,d)$ and $L_n = U_n(d,c,b,a)$. For $m,n \ge 3$, the complexity of $K_m^*L_n^*$ is $2^{m+n-1} - 2^{m-1} - 3 \cdot 2^{n-2} + 2$.
\end{theorem}

\begin{proof}
Let the various automata be defined as in the proof of Proposition~\ref{prop:KstarLstar}, but this time with $K_m=U_m(a,b,c,d)$ and $L_n=U_n(d,c,b,a)$.
The reachability of all of the states of $\cP$ follows the  proof in Theorem \ref{thm:KStarL} for all states $S \cup T$ where $n-1 \notin T$.
Let $T = \{0, t_1,\dots, t_k, n-1\}$. 
If $q_{m-1} \notin S$, then $S \cup \{0, t_2 - t_1, \dots, t_k - t_1, n-1 - t_1\} \der{d^{t_1}} S \cup T$.
If $q_{m-1} \in S$, say $S = \{q_0, q_{i_1}, \dots, q_{i_l}, q_{m-1}\}$, then $\{q_{i_1} - 1, \dots, q_{i_l - 1}, q_{m-2}\} \cup T \der{a} S \cup T$.
Therefore all allowable states are reachable.

For distinguishability, first consider states $S_1 \cup T_1$, $S_2 \cup T_2$,
where $S_1,S_2\subseteq Q_1$ and $T_1,T_2\subseteq \{s_2\} \cup Q_2$.
The set of final states of the NFA is $\{s_2, n-1\}$; however, any set containing $s_1$ or $q_{m-1}$ also contains $s_2$, and hence is a final state of $\cP$.
Note that applying $c$ always results in a state $S \cup T$, where $q_{m-1},s_2 \notin S$, and applying $b$ causes $n-1 \notin T$.
If $T_1 \neq T_2$, then applying a cyclic shift $d^k$ transforms the states so that $n-1 \in T_1 \oplus T_2$, and then applying $c$ distinguishing the states.
If $S_1 \neq S_2$, apply a cyclic shift $a^k$ so that $q_{m-1} \in S_1 \oplus S_2$, then apply $b$ to distinguish the states.

Finally, consider the initial state $\{s_1,s_2\}$, and any state $R$
not contain $s_1$, since the initial state is
the only one containing $s_1$. There are three cases:
\be
\item
$q_0 \not\in R$: Applying $a$,  from $\{s_1,s_2\}$ we reach $\{q_1,0\}$, and from $R$ we
reach $R'$, where $q_1 \not\in R'$. By the argument in the second 
paragraph of the proof, $\{s_1,s_2\}$ is distinguished from $R$.
\item
$q_0 \in R$, and $R \neq \{q_0,0\}$: If $ad$ is applied,  then $\{s_1,s_2\}$ goes to $\{q_1,1\}$, and $R$ goes to
$R'$ such that there exists $x \in R', x \not\in \{q_1,1\}$. Then these two
states are distinguishable by the previous argument.
\item
$R = \{q_0,0\}$: State $\{s_1,s_2\}$ is final, but $\{q_0,0\}$ is not.
\ee
Hence all the allowable states are distinguishable and the theorem holds.
\qed
\end{proof}
\section{Stars of Binary Operations}
\label{sec:starops}

\subsection{The Language $(KL)^*$}

In 2008 Gao, K. Salomaa, and Yu~\cite{GSY_FI08} proved that 
$2^{m+n-1}+2^{m+n-4}-(2^{m-1}+2^{n-1}-m-1)$ is a tight upper bound for $(KL)^*$. 
They used the following  DFA's over alphabet $\Sig=\{a,b,c,d\}$:
For $K$, let
$\cD_K=(Q_1,\Sig,\delta_K,q_0,\{q_{m-1}\})$ with 
$a:(q_0,\ldots,q_{m-1})$,
$b:\one_{Q_K}$,
$c$ defined by $\delta_K(q_0,c)=\delta_K(q_{m-1},c)=q_0$, $\delta_K(q_i,c)=q_{i+1}$, for $i=1,\ldots,m-2$, and 
$d:\one_{Q_K}$.
For $L$, let $\cD_L=(Q_L,\Sig,\delta_L,0,\{n-1\})$ with 
$a:\one_{Q_L}$,
$b:(0,\ldots,n-1)$, 
$c:\one_{Q_L}$, and 
$d$ defined by $\delta_L(0,d)=\delta_L(n-1,d)=0$, $\delta_L(i,d)=i+1$, for $i=1,\ldots,n-2$.
We show that two permutationally equivalent dialects $W_m(a,b,c,d)$ and $W_n(d,c,b,a)$ of $U_n(a,b,c,d)$ also meet the bound.

\begin{figure}[hbt]
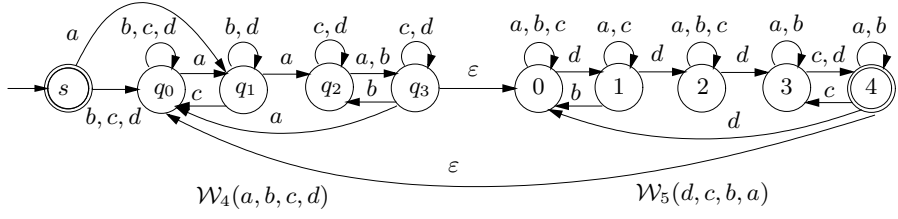

\begin{center}
\input prodstar.eepic
\end{center}
\caption{NFA for $((W_4(a,b,c,d)\;W_5(d,c,b,a))^*$.} 
\label{fig:prodstar}
\end{figure}

\begin{theorem}[$(KL)^*$]
Let $K_m=W_m(a,b,c,d)$ and $L_n=W_n(d,c,b,a)$. 
For $m,n\ge3$, the
complexity of  $(K_mL_n)^*$  is
$2^{m+n-1}+2^{m+n-4}-(2^{m-1}+2^{n-1}-m-1)$.
\end{theorem}
\begin{proof}
Let $\cD_1=(Q_1,\Sig,\delta_1,q_0,\{q_{m-1}\})$  with
$Q_1=\{q_0,\ldots,q_{m-1}\}$ be the DFA of $K_m$, and 
let  $\cD_2=(Q_2,\Sig,\delta_2,0,\{n-1\})$ with $Q_2=\{0,\ldots,n-1\}$ be the DFA of $L_n$.
Let $\cN$ be the NFA for $(KL)^*$.
This NFA  is shown in Fig.~\ref{fig:prodstar} for $m=4$ and $n=5$.
Let 
$\cD$ be the DFA obtained from $\cN$ by the subset construction and minimization.

The states of $\cD$ are the initial state $\{s\}$ and states of the form $S \cup T$ where $\emp \subsetneq S \subseteq Q_1$ and $T \subseteq Q_2$.
Because of the $\eps$-transitions, the allowable states $S \cup T$ must have either $q_{m-1} \notin S$ or $q_{m-1} \in S$, and $0 \in T$.
Moreover, if $|S| > 1$, then $T \neq \emp$, as at least one $\eps$-transition from $n-1$ to $q_0$ must have been used.
The number of allowable states is counted as follows:
\be
\item
First, we have the initial state $\{s\}$.
\item
If $T = \emptyset$, then $|S| = 1$, and $q_{m-1} \not\in S$. There are $m-1$
such states.
\item
If $ T \neq \emptyset$, then $|S| \ge 1$.
	\be
	\item
	 $n-1 \not\in T$: If $q_{m-1} \not\in S$, then there are
	$(2^{m-1}-1)(2^{n-1}-1)$ such states. Otherwise, $q_{m-1} \in S$ and $0 \in T$, and
	there are $2^{m+n-3}$ such states.
	\item
	$n-1 \in T$: Then $q_0 \in S$. If $q_{m-1} \not\in S$, there are
	$2^{m+n-3}$ such states. Otherwise, $q_{m-1} \in S$ and $0 \in T$, and there
	are $2^{m+n-4}$ such states.
	\ee
\ee
Altogether we have $2^{m+n-1} + 2^{m+n-4} -
(2^{m-1} + 2^{n-1} - m - 1)$ states.
We will now show they are all reachable.

The initial state is $\{s\}$. 
We have $\{s\} \der{b} \{q_0\} \der{a^i} \{q_i\}$ for $i < m-1$. 

For $i < m-1$ and $T = \{t_1,\dots, t_k\} \subseteq Q_2 \bs \{n-1\}$ with $t_1<\cdots<t_k$, the state $\{q_i\} \cup T$ is reachable by $\{q_i\} \cup \{t_2 - t_1,\dots, t_k - t_1\} \der{a^m d^{t_1}} \{q_i\} \cup T$.
Suppose $n-1 \in T$, say $T = \{t_1,\dots, t_k, n-1\}$.
If $T \neq Q_2$, then the state $\{q_0\} \cup T$ is reachable by a applying a cyclic shift $d^l$ to some $\{q_0\} \cup T'$, where $n-1 \notin T'$.
Moreover, $\{q_{m-2}\} \cup (Q_2 \bs \{n-1\}) \der{da} \{q_0,q_1,q_{m-1}\} \cup Q_2 \der{cac} \{q_0\} \cup Q_2$.
Finally, if $0 \in T$ then $\{q_{m-2}\} \cup T \der{a} \{q_{m-1}\} \cup T$.
So all allowable states of the form $S \cup T$, $|S| = 1$ are reachable.

Let $S = \{q_{i_1}, \dots, q_{i_k}\}$, $0 < i_1 < \cdots < i_k$. 
Since $n-1 \notin T$, we have $\{q_{i_2 - i_1}, \dots, q_{i_k - i_1}\} \cup T \der{d^na^{i_1}} S \cup T$.
Now suppose $S = \{q_0, q_{i_2}, \dots, q_{i_k}\}$.
If $n-1 \in T$, then $\{q_0, q_{i_3 - i_2}, \dots, q_{i_k - i_2}\} \cup T \der{a(ac^2)^{i_2 - 1}} S \cup T$.
If $n-1 \notin T$ and $q_{m-1} \in S$, then $T = \{0, t_2, \dots, t_l\}$ and $t_l < n-1$.
Let $T' = \{0, t_2 - 1, \dots, t_l - 1, n-1\}$. 
Then $S \cup T'$ is reachable, and $S \cup T' \der{d} S \cup T \cup \{1\}$; if $1 \notin T$, apply $b^2$ to get $S \cup T$.

Finally, suppose $q_0 \in S$, $q_{m-1}
\not\in S$,  and  $n-1 \not\in T$. Suppose $T = \{t_1,\ldots,t_l\}$, $t_1 < \cdots < t_l$,
and let $T'' = \{t_2 - t_1 - 1, \ldots, t_l - t_1 - 1, n-1\}$. Since $q_0 \in
S$ and $n-1 \in T$,  state $S \cup T''$ is reachable. Then we reach $S \cup
T$ from $S \cup T''$ by applying $d^{t_1 + 1}$.

Therefore all the allowable states are reachable.

We now show all states are disintinguishable.
Let $S_1 \cup T_1$, $S_2 \cup T_2$ be two distinct states.
If $T_1 \neq T_2$, then the states are distinguishable by a cyclic shift $d^k$.
If $S_1 \neq S_2$, without loss of generality we may assume $q_{m-1} \in S_1 \oplus S_2$.
Then applying $b^2d^{n-1}$ results in states $S_1' \cup T_1'$, $S_2' \cup T_2'$, where $0 \in T_1' \oplus T_2'$, so the states are distinguishable.
Finally, the initial state $\{s\}$ is distinguished from every state other than $\{q_0\}$ by  $a$;
it is distinguishable from $\{q_0\}$ because it is final.
\qed
\end{proof}

\subsection{The Languages $(K\cup L)^*$}

In 2007 A.\ Salomaa, K.\ Salomaa, and S.\ Yu~\cite{SSY07} showed that the complexity of 
$(K\cup L)^*$ is $2^{m+n-1} -(2^{m-1}+2^{n-1}-1)$ with ternary witnesses. 
Jir\'askov\'a and Okhotin~\cite{JiOk11} used
binary witnesses: 
For $K$, let
$\cD_K=(Q_1,\Sig,\delta_K,0,\{0\})$ with 
$a:(0,\ldots,{m-1})$, and 
$b$ defined by $\delta_K(i,b)=i+1$, for $i=0,\ldots,m-2$, $\delta_K(m-1,b)=1$.
For $L$, let $\cD_L=(Q_L,\Sig,\delta_L,0,\{0\})$ with 
$a:{0\choose 1}$  and
$b:(0,\ldots,n-1)$. 
Permutationally equivalent binary dialects of $U_n(a,b,c)$ can also be used.
Let $\cS_n=\cS_n(a,b)=(Q,\Sig,\delta_S, 0, \{0\})$, where  
$a:(0,\ldots,n-1)$, and $b:{0\choose 1}$.
The following theorem was proved in~\cite{Brz12}:
\begin{theorem}[$(K_m\cup L_n)^*$]\mbox{}
\noin
For $m,n \ge 3$,  the complexity of $(S_m(a,b)\cup S_n(b,a))^*$ is $2^{m+n-1} -(2^{m-1}+2^{n-1}-1)$.
\end{theorem}

\subsection{The Language $(K\cap L)^*$}
It was also proved in~\cite{JiOk11} that the complexity of $(K\cap L)^*$ is $2^{mn-1}+2^{mn-2}$, which is the composition of the complexities of intersection and star.
Their witnesses $K$ and $L$ were over an alphabet of six letters,  $\Sig=\{a,b,c,d,e,f\}$:
For $K$, let
$\cD_K=(Q_K,\Sig,\delta_K,0,\{m-1\})$, with $Q_K=\{0,\ldots,m-1\}$.
For $L$, let $\cD_L=(Q_L,\Sig,\delta_L,0,\{n-1\})$, with $Q_L=\{0,\ldots,n-1\}$.
The transitions were as follows:

\begin{equation*}
\begin{array}{ccc}
\cD_K & \quad \quad & \cD_L\\
\hline
a:(0,\ldots,m-1) & & a:(0,\ldots,n-1)\\
b:1_{Q_K} & & b:(0,\ldots,n-1)\\
c:(1,\ldots,m-1) & & c:1_{Q_L}\\ 
d:1_{Q_K} & & d:(1,\ldots,n-1)\\
e:{1\choose 0} & &  e:1_{Q_L}\\
 f:1_{Q_K} & & f:{1\choose 0}
\end{array}
\end{equation*}

We conjecture that quinary witnesses can also be used.
Let $\Sig=\{a,b,c,d,e\}$ and $\cU_n(a,b,c,d,e)=(Q_K,\Sig, \delta_\cU,0,\{n-1\})$,
where $Q_K=(0,\ldots,n-1\}$,
$a:(0,\ldots,n-1)$, $b:(0,1)$, $c:{n-1\choose 0}$, $d=1_{Q_K}$, and $e:(1,\ldots n-1)$.
Let $\pi$ be the permutation that sends $\{a,b,c,d,e\}$ to $\{e,c,b,a,d\}$,  let 
$\cD_1= \cU_n(a,b,c,d,e)$, and $\cD_2= \cU_n(e,c,b,a,d)$.
The transitions in $\cD_1$ and $\cD_2$ are:
\begin{equation*}
\begin{array}{ccc}
\cD_1 & \quad \quad & \cD_2\\
\hline
a:(0,\ldots,m-1) & & \pi(a):1_{Q_2}\\
b:(0,1) & & \pi(b):{n-1\choose 0}\\
c:{m-1\choose 0} & & \pi(c):(0,1)\\ 
d:1_{Q_1} & & \pi(d):(1,\ldots,n-1)\\
e:(1,\ldots,m-1)& &  \pi(e):(0,\ldots,n-1)
\end{array}
\end{equation*}
Note that $U_n(a,b,c,d,e)$ is an extension of $U_n(a,b,c,d)$ to 5 letters.
\begin{conj}[$(K_m\cap L_n)^*$]
\label{con:int}
Let $K_m=U_m(a,b,c,d,e)$ and $L_n=U_n(e,c,b,a,d)$.
Then the  complexity of $(K_m \cap L_n)^*$ is 
$2^{mn-1} +2^{mn-2}$ for $m,n \ge 3$. 
\end{conj}
This has been verified for $m=3$ and $n=3, 4, 5, 6$ and for $m=4$ and $n=4, 5$.

\subsection{The Language $(K\setminus L)^*$}

\begin{theorem}[$(K_m\setminus L_n)^*$]\mbox{}
\noin
The complexity of the operation $(K_m\setminus L_n)^*$
is $2^{mn-1} +2^{mn-2}$ for $m,n \ge 3$, and it is met by the witnesses $K_m$ and $\ol{L_n}$, where $K_m$ and $L_n$ are the witnesses of Jir\'askov\'a and Okhotin for intersection.
\end{theorem}
\begin{proof}

This follows since $(K\setminus \ol{L})^* =(K\cap L)^*$.
\qed
\end{proof}
If Conjecture~\ref{con:int} holds, then we also have 

\begin{conj}[$(K_m\setminus L_n)^*$]
Let $K_m=U_m(a,b,c,d,e)$ and $L_n=\ol{U_n(e,c,b,a,d)}$.
Then the  complexity of $(K_m \cap L_n)^*$ is 
$2^{mn-1} +2^{mn-2}$ for $m,n \ge 3$. 
\end{conj}

\subsection{The Language $(K\oplus L)^*$}
The complexity of this combined operation remains open.

\section{Conclusions}
\label{sec:conc}

We have proved that the universal witnesses $U_n(a,b,c)$ and  $U_n(a,b,c,d)$, along with their permutational equivalents $U_n(b,a,c)$ and $U_n(d,c,b,a)$, and dialects
$U_{\{0\},n}(a,b,c)$, $T_n(a,b,c)$, $T_n(b,a,c)$, $W_n(a,b,c,d)$, $W_{\{0\},n}(a,b,c,d)$,
$W_n(d,c,b,a)$ suffice to act as witnesses for all state complexity bounds involving binary boolean operations and product combined with star.
In the case of one or two starred arguments, 
we have shown that it is efficient to consider all four boolean operations together.
The use of universal witnesses and their dialects simplified several proofs, and allowed us to utilize the similarities in the witnesses.
\medskip

\noin 
{\bf Acknowledgment} We thank Baiyu Li for careful proofreading and correcting several flaws in an earlier version of the paper.

\end{document}